\documentclass[12pt,notitlepage]{article}
\usepackage{amsmath}
\usepackage{amssymb}
\usepackage{amscd}
\usepackage{amsfonts}
\usepackage{amsthm}
\usepackage{mathrsfs}
\usepackage[matrix,arrow,curve]{xy}
\usepackage{amsbsy}

\newcommand{\beq}[1]{\begin{equation}\label{#1}}
\newcommand{\eq}{\end{equation}}

  \newtheorem{theorem}{Theorem}[section]
  \newtheorem{df}[theorem]{Definition}
  \newtheorem{prop}[theorem]{Proposition}
  \newtheorem{lemma}[theorem]{Lemma}
  
\theoremstyle{remark}
  \newtheorem{ex}{Example}
  \newtheorem{rem}[theorem]{Remark}

\begin{document}

\begin{flushright}
ITEP-TH-32/19
\end{flushright}

\vspace{3cm}

\begin{center}
{\Large{\bf On the invariants of the full symmetric Toda system}}
\vspace{.5cm}
\ \\

Yu.B. Chernyakov\footnote{Institute for Theoretical and Experimental Physics, Bolshaya Cheremushkinskaya, 25,
117218 Moscow, Russia.}$^{,}$\footnote{Joint Institute for Nuclear Research, Bogoliubov Laboratory of Theoretical Physics, 141980 Dubna, Moscow region, Russia.}, chernyakov@itep.ru\\
G.I Sharygin\footnotemark[1]$^{,}$\footnotemark[2]$^{,}$\footnote{Lomonosov
Moscow State University, Faculty of Mechanics and Mathematics, GSP-1, 1 Leninskiye Gory, Main Building, 119991 Moscow, Russia.}, sharygin@itep.ru\\
A.S. Sorin\footnotemark[2]$^{,}$\footnote{National Research Nuclear University MEPhI
(Moscow Engineering Physics Institute),
Kashirskoye shosse 31, 115409 Moscow, Russia}$^{,}$\footnote{Dubna State University,
141980 Dubna (Moscow region), Russia.}, sorin@theor.jinr.ru\\

\end{center}

\begin{abstract}
In this paper we continue our study of the geometric properties of full symmetric Toda systems from \cite{CSS14,CSS17,CSS19}. Namely we describe here a simple geometric construction of a commutative family of vector fields on compact groups, that include the Toda vector field, i.e. the field, which generates the full symmetric Toda system associated with the Cartan decomposition of a semisimple Lie algebra. Our construction makes use of the representations of the semisimple algebra and does not depend on the splitness of the Cartan pair. It is very close to the family of invariants and semiinvariants of the Toda system associted with $SL_n$, introduced in \cite{CS}.
\end{abstract}

\medskip
\section{Introduction}
In this section we recall the basic constructions and definitions from Lie groups theory that we shall need in the rest of the paper. In particular, we define the Toda system and Toda fields here.

Let $\mathfrak g_\mathbb C$ be a complex semisimple Lie algebra, and $\mathfrak g$ its (non-compact) real form\footnote{From now on all the Lie algebras, Lie groups etc. in our paper will be real, unless otherwise stated and we shall omit the ground field in our notation; i.e. $SL_n=SL_n(\mathbb R),\,SO_n=SO_n(\mathbb R)$ etc.}; let $\theta:\mathfrak g\to\mathfrak g$ be a Cartan involution of $\mathfrak g$ and let
\[
\mathfrak g=\mathfrak k\oplus\mathfrak p
\]
be the corresponding Cartan decomposition, so that the following inclusions are true (where $[,]$ is the Lie bracket of $\mathfrak g$)
\begin{equation}
\label{eq:incl1}
[\mathfrak k,\mathfrak k]\subseteq\mathfrak k,\ [\mathfrak k,\mathfrak p]\subseteq\mathfrak p,\ [\mathfrak p,\mathfrak p]\subseteq\mathfrak k.
\end{equation}
In this case $\mathfrak k$ is a maximal compact subalgebra in $\mathfrak g$, and the Killing form $B$ of $\mathfrak g$ is negative-definite on $\mathfrak k$ and positive-definite on $\mathfrak p$. We shall denote by $G$ and $K\subset G$ the corresponding (1-connected) Lie groups.

Let $\mathfrak a$ be the maximal subalgebra in $\mathfrak p$; due to the inclusions \eqref{eq:incl1}, $\mathfrak a$ is always commutative. In case $\mathfrak a$ is a maximal commutative subalgebra in $\mathfrak g$ (with respect to inclusions), the real form $\mathfrak g$ is called \textit{normal} or \textit{split}. With $\mathfrak a$ fixed we can consider the root system of $\mathfrak g$ with respect to $\mathfrak a$; recall that an element $\alpha\in\mathfrak a^*$ is called \textit{root}, if the space
\[
\mathfrak g_\alpha=\{X\in\mathfrak g\mid\! [H,X]=\alpha(H)X,\,H\in\mathfrak a\}\subseteq\mathfrak g
\]
is non-zero. Let us fix a positive basis in $\mathfrak a^*$ and let $\Delta_+$ be the set of positive roots with respect to this basis. Then one can show that $\theta(\mathfrak g_\alpha)=\mathfrak g_{-\alpha}$ for all (positive or negative) roots of $\mathfrak g$. So that choosing the bases $\{e_\alpha^1,\,\dots,e_\alpha^k\}$ of $\mathfrak g_\alpha$ for positive roots $\alpha\in\Delta_+$ we obtain bases in $\mathfrak g_{-\alpha}$ by setting $e^i_{-\alpha}=\theta(e^i_\alpha)$; remark that in the split case $k=1$ so that there is no need to use the second index. In this notation we have
\[
\begin{aligned}
\mathfrak p&=\mathfrak a\oplus\bigoplus_{\alpha\in\Delta_+,i}\mathbb R(e^i_\alpha+e^i_{-\alpha}),\\
\mathfrak k&=\bigoplus_{\alpha\in\Delta_+,i}\mathbb R(e^i_\alpha-e^i_{-\alpha}).
\end{aligned}
\]
In the same way one defines the upper and lower Borel subalgebras in $\mathfrak g$:
\[
\mathfrak b_+=\mathfrak a\oplus\bigoplus_{\alpha\in\Delta_+,i}\mathbb Re^i_\alpha,\ \mathfrak b_-=\theta(\mathfrak b_+)=\mathfrak a\oplus\bigoplus_{\alpha\in\Delta_+,i}\mathbb R e^i_{-\alpha}.
\]
Using this notation, we define the \textit{Toda projection} $M:\mathfrak p\to\mathfrak k$ by the formula
\begin{equation}
\label{eq:Todaproj}
M\left(H+\sum_{\alpha\in\Delta_+,i}a^i_\alpha(e^i_\alpha+e^i_{-\alpha})\right)=\sum_{\alpha\in\Delta_+,i}a^i_\alpha(e^i_\alpha-e^i_{-\alpha}).
\end{equation}
(see \cite{dMP}). Observe that this map depends only on the choice of the positive base.

Now the following dynamics on $\mathfrak p$ is by definition called \textit{full symmetric Toda system} on $\mathfrak g$:
\begin{equation}\label{eq:Todaeq1}
\dot L=[M(L),L],\ L\in\mathfrak p.
\end{equation}
The right hand side of this formula is an element in $\mathfrak p$ due to the embeddings \eqref{eq:incl1}. The system \eqref{eq:Todaeq1} turns out to be Hamiltonian (the Poisson structure on $\mathfrak p$ being induced from its identification with $\mathfrak b_-^*$, given by the Killing form and the Hamilton function is given by $H(L)=Tr((ad_L)^2)$) and integrable in the sense that there exists a sufficient number of involutive integrals of this system on $\mathfrak p$; part of this family is the well-known Flaschke integrals $F_k(L)=Tr((ad_L)^k),\ k=1,2,\dots,\mathrm{rk}\,\mathfrak g$, the other part is given for instance by the so-called \textit{chopping procedure} (see \cite{DLNT}; see also \cite{CS2} for a detailed analysis of the number of independent integrals in case $G=SL_n$).

The existence of Flaschke integrals implies that any solution $L(T)$ of the \eqref{eq:Todaeq1} lies within a single orbit of the adjoint representation of $K$ on $\mathfrak p$; indeed consider the linear vector field $\tau_L=[M(L),-]$ on $\mathfrak p$, determined by this equation; then for every $L$, $\tau_L$ belongs to the subspace, spanned by the adjoint representation of $\mathfrak k$ at $L$. Thus we can in fact restrict the whole system to a fixed orbit $Ad_K(L)\subset\mathfrak p$ for some fixed $L\in\mathfrak p$; for instance since every orbit of $Ad_K$ in $\mathfrak p$ contains an element $\Lambda\in\mathfrak a$, we can take $L=\Lambda$ (observe that $\Lambda$ is unique up to an action of its stabilizer subgroup; in case of a generic element and normal real form this group is discrete).

Next we define the \textit{Toda vector fields} on the compact group $K$:
\begin{df}
\label{def:Todafield}
In the notations set here, let $\Lambda\in\mathfrak a$ be a generic element (in particular, we can assume that its centralizer group in $K$ is minimal possible). We put
\[
\mathscr T^\Lambda\in \Gamma^\infty(K,TK),\ \mathscr T^\Lambda(k)=dR_k(M(Ad_k(\Lambda))),
\]
where as usual $R_k$ denotes the right translation by the element $k\in K$ and $Ad_k(\Lambda)$ denotes the adjoint action of $k\in K$ on $\Lambda\in\mathfrak a\subset\mathfrak p$.
\end{df}
The relation of this field with equation \eqref{eq:Todaeq1} is given by the following proposition (see for example \cite{K1}, \cite{S}, \cite{Per}, \cite{TF}):
\begin{prop}\label{prop:flowt}
Let $L(t),\ t\in\mathbb R$ be the solution of equation \eqref{eq:Todaeq1} with the initial condition $L(0)=L_0\in\mathfrak p$. Let $L_0=Ad_{k_0}(\Lambda_0)$ for some $\Lambda_0\in\mathfrak a$ and $k_0\in K$. Consider the one-parameter family $k(t)\in K$, generated by the field $\mathscr T^{\Lambda_0}$ with $k(0)=k_0$. Then
\[
L(t)=Ad_{k(t)}(\Lambda_0).
\]
\end{prop}
\begin{proof}
Denote the 1-parameter family $Ad_{k(t)}(\Lambda_0)\in\mathfrak p$ by $\tilde L(t)$. Since $\tilde L(0)=L_0=L(0)$, it is enough to check that both families satisfy the same differential equation; so we compute:
\[
\dot{\tilde L}(t)=\frac{d}{dt}(Ad_{k(t)}(\Lambda_0))=[M(Ad_{k(t)}(\Lambda_0)),Ad_{k(t)}(\Lambda_0)]=[M(\tilde L(t)),\tilde L(t)].
\]
Here at the second equality we used the relation between the adjoint actions of Lie group and its algebra as well as the definition of the field $\mathscr T^{\Lambda_0}$.
\end{proof}
In what follows we shall use the properties of the Lie groups and algebras, as well as their representation theory to construct a large commutative family of vector fields on $K$, which will contain the Toda fields $\mathscr T^\Lambda$ for all $\Lambda\in\mathfrak a$. In the particular case $G=SL_n,\,K=SO_n$ our construction is closely related to the large family of invariants and semi-invariants of the system, constructed in \cite{CS}. Unfortunately, the method we use here makes it hard to perceive if the fields we construct are in fact related with some involutive family of functions on $\mathfrak p$; so we postpone the discussion of this question to a forthcoming paper (see however \cite{RS} for a similar geometric approach).

The remaining part of the paper is organized as follows: after presenting few preliminary constructions in section 2, we proceed with the main part of our construction in section 3, where we use the representation theory to modify the commutator relations. Section 4 contains a detailed illustration of our method in case $G=SL_4,\,K=SO_4$. And the final section of the paper contains some conjectures and discussion of relation between our results and what has been known earlier.

\section{Properties of Toda fields on $K$}
It is implied by the existence of commuting Flascke integrals that the Toda fields associated with different $\Lambda$ should commute. It turns out to be true and the method of proof yields some insight into the geometry of the system:
\begin{prop}\label{prop:commutsimple}
The fields $\mathscr T^{\Lambda_1},\,\mathscr T^{\Lambda_2}$ commute for all $\Lambda_1,\,\Lambda_2\in\mathfrak a$.
\end{prop}
\noindent To \textit{prove} this we begin with the following handy lemma: recall that any Lie group $G$ is a parallelizable manifold, i.e. its tangent bundle can be trivialized, one can use the left or right shifts by the group elements to identify all the tangent spaces $T_gG$ with the tangent space at the neutral (unit) element of the group, which is the Lie algebra $\mathfrak g$ of the $G$. To fix the notation we shall be using the right shift, so that every vector field has the form $X(g)=dR_g(x(g))$ for a suitable $\mathfrak g$-valued function $x:G\to\mathfrak g$; we shall say that \textit{$X(g)$ corresponds to $x(g)$}.
\begin{lemma}
\label{lem:formula1}
Let $X(g),\,Y(g)$ be two vector fields on $G$, corresponding to $x(g)$ and $y(g)$ respectively; then their commutator $[X,\,Y](g)$ corresponds to the function $z:G\to\mathfrak g$, given by the formula:
\[
z(g)=(\mathscr L_X(y))(g)-(\mathscr L_Y(x))(g)-[x(g),y(g)],
\]
where $\mathscr L_X(f)$ denotes the Lie derivative of a (vector-valued) function $f$ on $G$, and $[,]$ is the Lie bracket in $\mathfrak g$.
\end{lemma}
\begin{proof}
The easiest way to prove this formula is to check the properties of the expression on the right hand side: let $Z(X,Y)(g)=dR_g(z(g))$; clearly this expression determines a vector field on $G$, which depends bilinearly on $X$ and $Y$ (over constants). We also have:
\[
Z(X,Y)(g)=-Z(Y,X)(g),\ Z(X,fY)(g)=(\mathscr L_X(f))(g)Y(g)+Z(X,Y)(g),
\]
and
\[
Z(X,Y)=-dR_g([x,y])=[X,\,Y](g),
\]
if $x,\,y\in\mathfrak g$ regarded as constant functions $G\to\mathfrak g$ (i.e. if $X$ and $Y$ are right-invariant). The rest follows from the properties of the commutator of the vector fields and the fact that each vector field on $G$ can be represented as the sum $\sum f_i(g)X_i(g)$, where $X_i$ correspond to constant functions $x_i:G\to\mathfrak g$.
\end{proof}
Applying this lemma to $G=K$ and $X(k)=\mathscr T^{\Lambda_1}(k),\ Y(k)=\mathscr T^{\Lambda_2}(k)$ so that $x(k)=M(Ad_k(\Lambda_1)),\,y(k)=M(Ad_k(\Lambda_2))$, we have:
\[
\begin{aligned}
z(g)&=M([M(Ad_k(\Lambda_1)),Ad_k(\Lambda_2)])-M([M(Ad_k(\Lambda_2)),Ad_k(\Lambda_1)])\\
      &\quad-[M(Ad_k(\Lambda_1)),M(Ad_k(\Lambda_2))].
\end{aligned}
\]
Here we used the observation that the infinitesimal part of the adjoint group representation is given by the adjoint representation of the Lie algebra. Now the statement of the proposition follows from the next Lemma:
\begin{lemma}
The map $M$ verifies the following equation:
\[
[X,Y]=M([M(X),Y])+M([X,M(Y)])-[M(X),M(Y)]
\]
for all $X,Y\in\mathfrak p$.
\end{lemma}
\begin{proof}
One can prove this equality by direct computations in basis $e_{\pm\alpha}^i$, described above. For example if the real form is normal, so that there is no need to consider the indices $i$ let $X\in\mathfrak a,\,Y=e_\alpha+e_{-\alpha}$, then we have on the left $-\alpha(X)(e_\alpha-e_{-\alpha})$, and on the right
\[
\begin{aligned}
{}M([M(X),Y])&+M([X,M(Y)])-[M(X),M(Y)]\\
                    &=M([X,e_\alpha-e_{-\alpha}])=-M(\alpha(X)(e_\alpha+e_{-\alpha}))\\
                    &=-\alpha(X)(e_\alpha-e_{-\alpha}).
\end{aligned}
\]
If $Y$ is from $\mathfrak a$ and $X$ of the form $e_\alpha+e_{-\alpha}$, the computations are the same, since the formula is (anti)symmetric in $X,Y$. Finally, if $X=e_\alpha-e_{-\alpha},\,Y=e_\beta-e_{-\beta}$ and $\alpha=\beta$, then we have $0$ on both sides, because of the anti-symmetricity of the bracket; and if $\alpha\neq\beta$, then we have $\lambda_{\alpha\beta}(e_{\alpha+\beta}-e_{-\alpha-\beta})$ (for some structure constant $\lambda_{\alpha\beta}$) on the left, and on the right we have
\[
\begin{aligned}
{}M([M(X),Y])&+M([X,M(Y)])-[M(X),M(Y)]\\
                     &=2M(\lambda_{\alpha\beta}(e_{\alpha+\beta}+e_{-\alpha-\beta}))-\lambda_{\alpha\beta}(e_{\alpha+\beta}-e_{-\alpha-\beta})\\
                     &=\lambda_{\alpha\beta}(e_{\alpha+\beta}-e_{-\alpha-\beta}).
\end{aligned}
\]
In the general case computations are pretty much the same.
\end{proof}
Now to finish the proof of the proposition \ref{prop:commutsimple} it is enough to observe that $Ad_k$ is a Lie algebra homomorphism and $[\Lambda_1,\Lambda_2]=0$ since the algebra $\mathfrak a$ is commutative. \hfill$\square$

\medskip
\noindent
In effect, this construction is a reflection of a bit more general situation: in fact we can regard the map $M$ as a projection $\mathfrak g\to\mathfrak k$, corresponding to the direct sum decomposition
\begin{equation}
\label{eq:decomp1}
\mathfrak g=\mathfrak k\oplus \mathfrak b_-.
\end{equation}
More accurately, the map $M$ we considered above is just the restriction of this projection to $\mathfrak p$.

Since both summands in the decomposition \eqref{eq:decomp1} are Lie subalgebras in $\mathfrak g$, the projector $M$ verifies the so-called \textit{Nijenhuis equation}:
\[
M^2([X,Y])-M([M(X),Y])-M([X,M(Y)])+[M(X),M(Y)]=0.
\]
In previous lemma, we wrote $[X,Y]$ instead of $M^2([X,Y])$, since $M^2=M$ it being a projector, and besides this $[X,Y]\in\mathfrak k$ according to the \eqref{eq:incl1} where $\mathfrak k$ is the image of $M$. To prove this equation (up to our knowledge, first observed by Kosmann-Schwarzbach and Magri, see \cite{KScM}), we write the arguments as $X=X_1+X_2,\,Y=Y_1+Y_2$ according to the decomposition \eqref{eq:decomp1} and compute:
\[
\begin{aligned}
M^2([X,Y])&=[X_1,Y_1]+M([X_1,Y_2])+M([X_2,Y_1]),\\
M([M(X),Y])&=[X_1,Y_1]+M([X_1,Y_2]),\\
M([X,M(Y)])&=[X_1,Y_1]+M([X_2,Y_1]),\\
[M(X),M(Y)]&=[X_1,Y_1].
\end{aligned}
\]
So summing up the first and the fourth equations in this formula and subtracting the second and the third we get the necessary result.

Now we can consider the vector fields $\mathscr T^X$ for all $X\in\mathfrak g$ given by similar formula:
\begin{equation}\label{eq:Todafield2}
\mathscr T^X(k)=dR_k(M(Ad_k(X))).
\end{equation}
We shall call these fields \textit{generalised Toda fields}. In order to compute the commutator $[\mathscr T^X,\mathscr T^Y]$ we shall use the formula from lemma \ref{lem:formula1}: we have $x(k)=M(Ad_k(X)),\, y(k)=M(Ad_k(Y))$, so reasoning as above we obtain
\[
[\mathscr T^X,\mathscr T^Y](k)=dR_k(z(k)),
\]
for
\[
\begin{aligned}
z(k)&=M([M(Ad_k(X)),Ad_k(Y)])+M([Ad_k(X),M(Ad_k(Y))])\\
      &\quad-[M(Ad_(X)),M(Ad_k(Y))]\\
      &=M^2([Ad_k(X),Ad_k(Y)])=M(Ad_k([X,Y])).
\end{aligned}
\]
So, $[\mathscr T^X,\mathscr T^Y]=\mathscr T^{[X,Y]}$ and we obtain \textit{a representation of $\mathfrak g$ in vector fields on $K$}.
\begin{ex}
Let $\mathfrak g=\mathfrak{sl}_2,\ G=SL_2$, then the maximal compact subgroup in $G$ is $K=SO_2\cong S^1$. The construction we just described gives a representation of $\mathfrak{sl}_2$ by vector fields on the circle. To describe this representation, we choose the basis $E,\,F,\,H$ of $\mathfrak{sl}_2$:
\[
\begin{aligned}
E&=\begin{pmatrix}0 & 1\\ 0 & 0\end{pmatrix}, & F&=\begin{pmatrix}0 & 0\\ 1 & 0\end{pmatrix}, & H&=\begin{pmatrix}1 & 0\\ 0 & -1\end{pmatrix},
\end{aligned}
\]
so that we have the following commutator relations:
\[
\begin{aligned}
{} [H,E]&= 2E, & [H,F]&=-2F, & [E,F]&=H.
\end{aligned}
\]
On the other hand, any vector field on $S^1$ can be written as $\xi=f(\alpha)\partial_\alpha$, where $\alpha$ is the rotation angle, $\partial_\alpha=\frac{\partial}{\partial\alpha}$ and $f$ is a $2\pi$-periodic function. Now  straightforward computation gives the following representation of the Lie algebra $\mathfrak{sl}_2$ by vector fields:
\[
\begin{aligned}
\mathscr T^E&=-\cos^2{\!\alpha}\,\partial_\alpha, & \mathscr T^F&=-\sin^2{\!\alpha}\,\partial_\alpha, & \mathscr T^H&=-\sin{2\alpha}\,\partial_\alpha.
\end{aligned}
\]
Indeed the commutation relation for these fields are
\[
\begin{aligned}
{}[\mathscr T^H,\mathscr T^E]&=[\sin2\alpha\,\partial_\alpha,\cos^2\!\alpha\,\partial_\alpha]=-2\cos^2\!\alpha\,\partial_\alpha=2\mathscr T^E,\\
{}[\mathscr T^H,\mathscr T^F]&=[\sin2\alpha\,\partial_\alpha,\sin^2\!\alpha\,\partial_\alpha]=2\sin^2\!\alpha\,\partial_\alpha=-2\mathscr T^F,\\
{}[\mathscr T^E,\mathscr T^F]&=[\cos^2\!\alpha\,\partial_\alpha,\sin^2\!\alpha\,\partial_\alpha]=-\sin2\alpha\,\partial_\alpha=\mathscr T^H.
\end{aligned}
\]
\end{ex}
\begin{rem}
There is a geometric way to think about fields $\mathscr T^X$ if the real form $\mathfrak g$ is split: consider the flag space $Fl(G)=G/B_-$ (here $B_-$ is the Borel subgroup, corresponding to subalgebra $\mathfrak b_-$). As one knows, one can identify it with the quotient space of the maximal compact subgroup $K\subset G$ by the action of the group $N=K\bigcap B_-$:
\[
Fl(G)=K/(K\bigcap B_-).
\]
Now the intersection $K\bigcap B_-$ is discrete if the real form is norma, so the natural projection $K\to Fl(G)$ is local diffeomorphism. The group $G$ acts on its quotient space $Fl(G)$ by left translations, so the infinitesimal part of this action determines the representation of $\mathfrak g$ by fields $\tilde{\mathscr T}^X,\,X\in\mathfrak g$ on $Fl(G)$. Since the group $K$ is locally diffeomorphic to $Fl(G)$, we can raise $\tilde{\mathscr T}^X$ to some vector fields on $K$, verifying the same commutator relations; thus we obtain a representation of $\mathfrak g$ on $K$ by vector fields. It is easy to see, that the fields $\mathscr T^X$ correspond to $\tilde{\mathscr T}^X$ under this construction.
\end{rem}

\section{The commutative family}
The construction of representation $\mathscr T$ is quite useful. In particular, it gives a commutative family of vector fields on $K$ for every commutative subalgebra of $\mathfrak g$ (and not only for the algebra $\mathfrak a$). However, these more general fields need not commute with the Toda fields $\mathscr T^\Lambda,\ \Lambda\in\mathfrak a$. In order to obtain a commutative subalgebra of vector fields, that would include the Toda fields, let us consider the following construction.

\begin{rem}
For the sake of avoiding redundant indices, we shall write $e_\alpha$ for a root vector in $\mathfrak g$; however, all the construction described in this section can be repeated word by word for any non-split real form of a semisimple Lie algebra.
\end{rem}

We are going to change the fields $\mathscr T^X$ by multiplying them by a suitable function on $K$ so that the resulting field will commute with the Toda field $\mathscr T^\Lambda$. To this end we choose a (finite-dimensional) bottom-weight representation $\rho:\mathfrak g\to End(V)$ of $\mathfrak g$ (this notion is dual to that of top-weight representation), and let $v_-=v_1$ be the lowest (bottom weight) vector of this representation, i.e. when one passes to the root vectors basis in $\mathfrak g$, one has:
\[
\rho(e_{-\alpha})(v_-)=0,\ \rho(h)(v_-)=\omega_1(h)v_-
\]
for all negative root vectors $e_{-\alpha}\in\mathfrak g,\ \alpha\in\Delta_+$ and $h\in\mathfrak a$ (here $\omega_1\in\mathfrak a^*$ is the lowest weight of the representation). The purpose of choosing bottom- instead of top-weight representations is that in our construction we project along $\mathfrak b_-$ and not $\mathfrak b_+$, so we shall need to make sure that the elements from this algebra will not interfere with our computations (see equation \eqref{eq:compute1}).

Let $v_-=v_1,\dots,v_N$ be the basis of $V$; we can assume that this basis is orthonormal with respect to a $K$-invariant Euclidean structure $\langle,\rangle$ on $V$ (i.e. invariant with respect to the induced representation $\hat\rho$ of the maximal compact group $K$ inside the one-connected group $G$, associated with $\mathfrak g$) and diagonal with respect to the action of the subalgebra $\mathfrak a$, i.e. that
\beq{weights}
\rho(h)(v_i)=\omega_i(h)v_i,\ \omega_i\in\mathfrak a^*,\ i=1,\dots,N.
\eq
The invariance of $\langle,\rangle$ means in particular that
\begin{equation}\label{eq:invscalar}
\langle \rho(X)(v),w\rangle=\langle v,\rho(\theta(X))(w)\rangle,\ \mbox{for all}\ X\in\mathfrak g,\ v,w\in V.
\end{equation}
In particular \[\langle\rho(L)(v),w\rangle=\langle v,\rho(L)(w)\rangle\] for all $L\in\mathfrak p$ and \[\langle\rho(O)(v),w\rangle=-\langle v,\rho(O)(w)\rangle\] for all $O\in\mathfrak k$; similar equations will hold for the induced group representations.

Consider now $N$ smooth functions on $K$
\beq{function1}
F^\rho_i=\langle \hat\rho(k)(v_i),v_-\rangle,\ k\in K.
\eq
To put this simply, these functions are just the bottom row of the matrix of $\hat\rho(k)$ in the basis $v_0,\dots,v_N$. Then, since $\mathscr T^\Lambda(k)=dR_k(M(Ad_k(\Lambda))$, we have the following formula:
\[
\mathscr T^\Lambda(F_i^\rho)(k)=\langle \rho(M(Ad_k(\Lambda))(\hat\rho(k)(v_i)),v_-\rangle=-\langle\hat\rho(k)(v_i),\rho(M(Ad_k(\Lambda))(v_-)\rangle
\]
(the last equality follows from the invariance of the scalar product, see \eqref{eq:invscalar} and the formulas after it). On the other hand, it follows from the definition, that
\[
\begin{aligned}
\rho(M(Ad_k(\Lambda))(v_-)&=\rho(Ad_k(\Lambda))(v_-)-\rho(Ad_k(\Lambda)_\mathfrak h)(v_-)\\&=\rho(Ad_k(\Lambda))(v_-)-\omega_0(Ad_k(\Lambda)_\mathfrak h)v_-,
\end{aligned}
\]
where $Ad_k(\Lambda)_\mathfrak h$ is the projection of $Ad_k(\Lambda)$ on $\mathfrak a$ along the root subspaces; indeed since the vector $v_0=v_-$ is killed by every $e_{-\alpha},\ \alpha\in\Delta_+$, the difference between $\rho(Ad_k(\Lambda))(v_-)$ and $\rho(M(Ad_k(\Lambda)))(v_-)$ consists of $\rho(Ad_k(\Lambda)_\mathfrak h)(v_-)$. We shall denote the function $\omega_0(Ad_k(\Lambda)_\mathfrak h)$ by $g_{\rho,\Lambda}(k)$; it is a smooth function on the compact group $K$, hence it is bounded; observe that it depends on $\rho$ and $\Lambda$, but does not depend on $i=0,\dots,N$. Then, using this notation we have
\begin{equation}\label{eq:compute1}
\begin{aligned}
\mathscr T^\Lambda(F_i^\rho)(k)&=-\langle\hat\rho(k)(v_i),\rho(Ad_k(\Lambda))(v_-)\rangle+\omega_0(Ad_k(\Lambda)_\mathfrak h)\langle\hat\rho(k)(v_i),v_-\rangle\\
                                &=-\langle\hat\rho(k)(\rho(\Lambda)(v_i)),v_-\rangle+g_{\rho,\Lambda}(k)F^\rho_i(k)\\
                                &=(-\omega_i(\Lambda)+g_{\rho,\Lambda}(k))F^\rho_i(k).
\end{aligned}
\end{equation}
In other words, the function $F_i^\rho$ has everywhere bounded logarithmic derivative along $\mathscr T^\Lambda$. In particular, its zero surfaces are invariant surfaces of the Toda flow. These surfaces are a straightforward generalizations of the minor surfaces of \cite{CS}.

On the other hand, it is easy to see that the ratio
\beq{function2}
F^\rho_{\omega_j-\omega_i}=\frac{F_i^\rho}{F_j^\rho}
\eq
verifies the equality $\mathscr T^\Lambda(F^\rho_\alpha)=\alpha(\Lambda)F^\rho_\alpha$, where $\alpha$ is the root of $\mathfrak g$, corresponding to the difference of weights $\omega_j-\omega_i$. Combining these functions for different $\rho$ and the same $\alpha$ we can obtain functions, invariant with respect to the vector field $\mathscr T^\Lambda$ (see section 4). Similarly, we can consider the product $\mathscr T^\rho_\alpha=\frac{1}{F_\alpha^\rho}\mathscr T^{e_\alpha}=F_{-\alpha}^\rho\mathscr T^{e_\alpha}$, then
\[
[\mathscr T^\Lambda,\mathscr T^\rho_\alpha]=0.
\]
In other words, \textit{for all $\rho$ and $\alpha$ the vector fields $\mathscr T^\rho_\alpha$ commute with the Toda fields $\mathscr T^\Lambda$}. In generic case however they need not commute with each other. In order to change this, consider the action of the field $\mathscr T^X$ on $F_i^\rho$ for arbitrary $X\in\mathfrak g$. Reasoning just as above we obtain
\[
\mathscr T^X(F^\rho_i)(k)=g_{\rho,X}(k)F^\rho_i(k)-\langle\hat\rho(k)(\rho(\theta(X))(v_i)),v_-\rangle,
\]
where $g_{\rho,X}(k)$ is defined similarly to $g_{\rho,\Lambda}(k)$ and does not depend on $i$; the second term has this form because of the invariance of the scalar product, see \eqref{eq:invscalar}. In particular, if $\rho(\theta(X))(v_i)=0$ we have $\mathscr T^X(F^\rho_i)(k)=g_{\rho,X}(k)F^\rho_i(k)$. If the same holds for $v_j$, then we have $\mathscr T^X(F^\rho_{\omega_i-\omega_j})=0$.

Thus, if for the roots $\alpha,\beta\in\Delta_+$ such that $[e_\alpha,e_\beta]=0$ we find the representations $\rho_\alpha,\rho_\beta$ (on spaces $V_\alpha,\,V_\beta$ respectively) and choose the basis vectors $v_i,v_j\in V_\alpha,\,v_k,v_l\in V_\beta$ such that the differences of the corresponding weights are equal to $\alpha$ and $\beta$ and
\[
\rho_\alpha(\theta(e_\beta))(v_i)=\rho_\alpha(\theta(e_\beta))(v_j)=\rho_\beta(\theta(e_\alpha))(v_k)=\rho_\beta(\theta(e_\alpha))(v_l)=0,
\]
then $\mathscr T^{e_\alpha}(F^{\rho_\beta}_\beta)=\mathscr T^{e_\beta}(F^{\rho_\alpha}_\alpha)=0$. So we have
\[
\begin{aligned}
{}[\mathscr T^{\rho_\alpha}_\alpha&,\mathscr T^{\rho_\beta}_\beta]=[F^{\rho_\alpha}_\alpha\mathscr T^{e_\alpha},F^{\rho_\beta}_\beta\mathscr T^{e_\beta}]\\
                                                       &=F^{\rho_\alpha}_\alpha\mathscr T^{e_\alpha}(F^{\rho_\beta}_\beta)\mathscr T^{e_\beta}-F^{\rho_\beta}_\beta\mathscr T^{e_\beta}(F^{\rho_\alpha}_\alpha)\mathscr T^{e_\alpha}+F^{\rho_\alpha}_\alpha F^{\rho_\beta}_\beta[\mathscr T^{e_\alpha},\mathscr T^{e_\beta}]=0.
\end{aligned}
\]
At the same time, we have the relations
\[
\begin{aligned}
{}[\mathscr T^\Lambda,\mathscr T^{\rho_\alpha}_\alpha]&=0,\\
{}[\mathscr T^\Lambda,\mathscr T^{\rho_\beta}_\beta]&=0.
\end{aligned}
\]
So in this way we obtain a large commutative algebra of vector fields, commuting with the Toda fields.

An interesting question is whether the fields $\mathscr T^X$ and more generally their combinations such as $\mathscr T^\rho_\alpha$ can be related to some Hamiltonian functions on $\mathfrak p$. We shall address this question in a forthcoming paper.

\section{Representations $\mathfrak{sl_4}$}

In this section we apply the construction, given earlier to the algebra $\mathfrak{g=sl_4}$ and the group $K=SO_4$, more accurately, we give explicit computations of the functions $F_\alpha^\rho$ in terms of the standard coordinate functions on group $SO_4$; it turns out that these functions coincide with the earlier defined in \cite{CS2} ratios of minors.

This real form is split and the rank of $\mathfrak{sl_4}$ is equal to $3$ (the dimension of its Cartan subalgebra equal to $3$), so we have three simple roots $\alpha_{1}, \alpha_{2}, \alpha_{3}$. Fundamental representation of $\mathfrak{sl_4}$ (i.e. the tautological one) is 4-dimensional, we denote by $V=\mathbb R^4$ the space of representation; in this terms other top weight irreducible representations are 6-dimensional on space $\wedge^{2} V$ and 4-dimensional on space $\wedge^{3} V\cong V^*$. Also observe that in this case the notions of top- and bottom-weight representations coincide and due to this we can talk about top vectors rather than bottom vectors. Besides this there is duality between $V$ and $V^*$, so matrix elements of these two representations are related by transposition.
\paragraph{Tautological 4-dimensional representation on $V$}
The root vectors and the elements of Cartan subalgebra in the tautological representation are given by the following matrices:
$$
\small{
\begin{array}{c}
E_{\alpha_{1}}=
\left(
\begin{array}{cccc}
 0 & 1 & 0 & 0 \\
 0 & 0 & 0 & 0 \\
 0 & 0 & 0 & 0 \\
 0 & 0 & 0 & 0
\end{array}
\right), \
E_{\alpha_{2}}=
\left(
\begin{array}{cccc}
 0 & 0 & 0 & 0 \\
 0 & 0 & 1 & 0 \\
 0 & 0 & 0 & 0 \\
 0 & 0 & 0 & 0
\end{array}
\right), \
E_{\alpha_{3}}=
\left(
\begin{array}{cccc}
 0 & 0 & 0 & 0 \\
 0 & 0 & 0 & 0 \\
 0 & 0 & 0 & 1 \\
 0 & 0 & 0 & 0
\end{array}
\right),\\
\ \\
E_{\alpha_{1}+\alpha_{2}}=
\left(
\begin{array}{cccc}
 0 & 0 & 1 & 0 \\
 0 & 0 & 0 & 0 \\
 0 & 0 & 0 & 0 \\
 0 & 0 & 0 & 0
\end{array}
\right), \
E_{\alpha_{2}+\alpha_{3}}=
\left(
\begin{array}{cccc}
 0 & 0 & 0 & 0 \\
 0 & 0 & 0 & 1 \\
 0 & 0 & 0 & 0 \\
 0 & 0 & 0 & 0
\end{array}
\right),\\
\ \\
E_{\alpha_{1}+\alpha_{2}+\alpha_{3}}=
\left(
\begin{array}{cccc}
 0 & 0 & 0 & 1 \\
 0 & 0 & 0 & 0 \\
 0 & 0 & 0 & 0 \\
 0 & 0 & 0 & 0
\end{array}
\right),\\
\ \\
h_{\alpha_{1}}=
\left(
\begin{array}{cccc}
 1 & 0 & 0 & 0 \\
 0 & -1 & 0 & 0 \\
 0 & 0 & 0 & 0 \\
 0 & 0 & 0 & 0
\end{array}
\right), \
h_{\alpha_{2}}=
\left(
\begin{array}{cccc}
 0 & 0 & 0 & 0 \\
 0 & 1 & 0 & 0 \\
 0 & 0 & -1 & 0 \\
 0 & 0 & 0 & 0
\end{array}
\right), \
h_{\alpha_{3}}=
\left(
\begin{array}{cccc}
 0 & 0 & 0 & 0 \\
 0 & 0 & 0 & 0 \\
 0 & 0 & 1 & 0 \\
 0 & 0 & 0 & -1
\end{array}
\right).
\end{array}
}
$$
In what follows we shall often omit the superscript $\rho$ from our notation when it does not cause confusion; in this terms we are looking for the functions
defined by the formulas (\ref{function1}) and (\ref{function2}). First we must find the correspondance between the weights and the roots. Let's remind that the weights $\omega_{i}$ and the roots $\alpha_{i}$ belong to the dual space $\mathfrak{a}^{\ast}$, in our case it is 3-dimensional. The weights are defined by the the formula (\ref{weights}):
$$
\rho_{4}(h)v_{i} = h_{\alpha_{j}} v_{i} = \omega_{i}(h_{\alpha_{j}}) v_{i}
$$
Then
$$
\small{
\begin{array}{c}
h_{\alpha_{1}}v_{1} = v_{1}, \
h_{\alpha_{2}}v_{1} = 0, \
h_{\alpha_{3}}v_{1} = 0,\\
\ \\
h_{\alpha_{1}}v_{2} =-v_{2}, \
h_{\alpha_{2}}v_{2} =v_{2}, \
h_{\alpha_{3}}v_{2} =0,\\
\ \\
h_{\alpha_{1}}v_{3} =0, \
h_{\alpha_{2}}v_{3} =-v_{3}, \
h_{\alpha_{3}}v_{3} =v_{3},\\
\ \\
h_{\alpha_{1}}v_{4} =0, \
h_{\alpha_{2}}v_{4} =0, \
h_{\alpha_{3}}v_{4} =-v_{4},
\end{array}
}
$$
and
$$
\small{
\begin{array}{c}
\omega_{1}=
\left(
\begin{array}{c}
 1\\
 0\\
 0
\end{array}
\right), \
\omega_{2}=
\left(
\begin{array}{c}
 -1\\
 1\\
 0
\end{array}
\right), \
\omega_{3}=
\left(
\begin{array}{c}
 0\\
 -1\\
 1
 \end{array}
\right), \
\omega_{4}=
\left(
\begin{array}{c}
 0\\
 0\\
-1
\end{array}
\right).
\end{array}
}
$$
Taking into account that
$$
\small{
\begin{array}{c}
[h, E_{\alpha_{i}}]=\alpha_{i}(h)E_{\alpha_{i}},\\
\ \\
v_{-}=v_{1}=\left(
\begin{array}{c}
 1\\
 0\\
 0\\
 0
\end{array}
\right), \ \hat{\rho}(k)=\left(
\begin{array}{cccc}
 \psi _{11} & \psi _{12} & \psi _{13} & \psi _{14} \\
 \psi _{21} & \psi _{22} & \psi _{23} & \psi _{24} \\
 \psi _{31} & \psi _{32} & \psi _{33} & \psi _{34} \\
 \psi _{41} & \psi _{42} & \psi _{43} & \psi _{44}
\end{array}
\right) \in SO(4),
\end{array}
}
$$

\noindent we get all functions $F_{j}$ and $F_{\alpha_{ij}}$ using the formulas (\ref{function1}) and (\ref{function2})
$$
\begin{array}{c}
F_{j}=\psi_{kl}, \ \
\psi^{'}_{kl}=(-\lambda_{l}+a_{11})M_{kl}.
\end{array}
$$
Their dynamics is expressed by the formula (\ref{eq:compute1}). Finally, we get the results and arrange them in the tables below
\begin{center}
\mbox{Table 1}\\
\ \\
\small{
\begin{tabular}{|c|c|c|c|c|}
\hline\cline{1-0}
$$ & $$ & $$ & $$ & $$\\
$$  & $F_{1}$ & $F_{2}$ & $F_{3}$ & $F_{4}$\\

$$ & $$ & $$ & $$ & $$\\
\hline\cline{1-0}
$$ & $$ & $$ & $$ & $$\\
$\psi$ & $\psi_{11}$ & $\psi_{12}$ & $\psi_{13}$ & $\psi_{14}$\\

$$ & $$ & $$ & $$ & $$\\
\hline\cline{1-0}
$$ & $$ & $$ & $$ & $$\\
$\omega_{i} \in \mathfrak{a}^{\ast}$ & $\left(
\begin{array}{c}
 1\\
 0\\
 0
 \end{array}
\right)$ & $\left(
\begin{array}{c}
 -1\\
 1\\
 0
 \end{array}
\right)$  & $\left(
\begin{array}{c}
 0\\
 -1\\
 1
 \end{array}
\right)$ & $\left(
\begin{array}{c}
 0\\
 0\\
 -1
 \end{array}
\right)$\\

$$ & $$ & $$ & $$ & $$\\
\hline\cline{1-0}
$$ & $$ & $$ & $$ & $$\\
$\omega_{i}(\Lambda)$ & $\lambda_{1}$ & $\lambda_{2}$ & $\lambda_{3}$ & $\lambda_{4}$\\
$$ & $$ & $$ & $$ & $$\\

\hline
\end{tabular}
}
\end{center}

\newpage
\begin{center}
\mbox{Table 2}\\
\ \\
\small{
\begin{tabular}{|c|c|c|c|}
\hline\cline{1-0}
$$ & $$ & $$ & $$\\
$\alpha$  & $\alpha \in \mathfrak{a}^{\ast}$ & $\omega_{i}-\omega_{j}$ & $F_{\alpha_{ij}=\omega_{i}-\omega_{j}}=F^{i}/F^{j}$\\

$$ & $$ & $$ & $$\\
\hline\cline{1-0}
$$ & $$ & $$ & $$\\
$\alpha_{1}$ & $\left(
\begin{array}{c}
 2\\
 -1\\
 0
 \end{array}
\right)$ & $\omega_{1}-\omega_{2}$ & $\psi_{11}/\psi_{12}$\\

$$ & $$ & $$ & $$\\
\hline\cline{1-0}
$$ & $$ & $$ & $$\\
$\alpha_{2}$ & $\left(
\begin{array}{c}
 -1\\
 2\\
 -1
 \end{array}
\right)$ & $\omega_{2}-\omega_{3}$ & $\psi_{12}/\psi_{13}$\\
$$ & $$ & $$ & $$\\
\hline\cline{1-0}
$$ & $$ & $$ & $$\\
$\alpha_{3}$ & $\left(
\begin{array}{c}
 0\\
 -1\\
 2
 \end{array}
\right)$ & $\omega_{3}-\omega_{4}$ & $\psi_{13}/\psi_{14}$\\
$$ & $$ & $$ & $$\\
\hline\cline{1-0}
$$ & $$ & $$ & $$\\
$\alpha_{1}+\alpha_{2}$ & $\left(
\begin{array}{c}
 1\\
 1\\
 -1
 \end{array}
\right)$ & $\omega_{1}-\omega_{3}$ & $\psi_{11}/\psi_{13}$\\
$$ & $$ & $$ & $$\\
\hline\cline{1-0}
$$ & $$ & $$ & $$\\
$\alpha_{2}+\alpha_{3}$ & $\left(
\begin{array}{c}
 -1\\
 1\\
 1
 \end{array}
\right)$ & $\omega_{2}-\omega_{4}$ & $\psi_{12}/\psi_{14}$\\
$$ & $$ & $$ & $$\\
\hline\cline{1-0}
$$ & $$ & $$ & $$\\
$\alpha_{1}+\alpha_{2}+\alpha_{3}$ & $\left(
\begin{array}{c}
 1\\
 0\\
 1
 \end{array}
\right)$ & $\omega_{1}-\omega_{4}$ & $\psi_{11}/\psi_{14}$\\

\hline
\end{tabular}
}
\end{center}

\newpage
\paragraph{4-dimensional representation on $\wedge^{3} V$.}
We choose the following basis in our representation
$$
\begin{array}{c}
u_{1}=e_{1} \wedge e_{2} \wedge e_{3}, \ u_{2}=e_{1} \wedge e_{2} \wedge e_{4}, \ u_{3}=e_{1} \wedge e_{3} \wedge e_{4}, \ u_{4}=e_{2} \wedge e_{3} \wedge e_{4}.
\end{array}
$$
This representation is then isomorphic to the dual of the tautological one; in particular, using the relations between the minors of orthogonal matrices we obtain the following equalities:
$$
\begin{array}{c}
F_{j}=\psi_{kl}, \ \ \
\psi^{'}_{kl}=(\lambda_{l}-a_{44})\psi_{kl},\\
\end{array}
$$
Then, the results can be arranged in the tables below
\begin{center}
\mbox{Table 3}\\
\ \\
\scriptsize{
\begin{tabular}{|c|c|c|c|}
\hline\cline{1-0}
$$ & $$ & $$ & $$\\
$\alpha$  & $\alpha \in \mathfrak{a}^{\ast} $ & $\omega_{i}-\omega_{j}$ & $F_{\alpha_{ij}=\omega_{i}-\omega_{j}}=F^{i}/F^{j}$\\

$$ & $$ & $$ & $$\\
\hline\cline{1-0}
$$ & $$ & $$ & $$\\
$\alpha_{1}$ & $\left(
\begin{array}{c}
 2\\
 -1\\
 0
 \end{array}
\right)$ & $\omega_{1}-\omega_{2}$ & $\psi_{41}/\psi_{42}$\\

$$ & $$ & $$ & $$\\
\hline\cline{1-0}
$$ & $$ & $$ & $$\\
$\alpha_{2}$ & $\left(
\begin{array}{c}
 -1\\
 2\\
 -1
 \end{array}
\right)$ & $\omega_{2}-\omega_{3}$ & $\psi_{42}/\psi_{43}$\\
$$ & $$ & $$ & $$\\
\hline\cline{1-0}
$$ & $$ & $$ & $$\\
$\alpha_{3}$ & $\left(
\begin{array}{c}
 0\\
 -1\\
 2
 \end{array}
\right)$ & $\omega_{3}-\omega_{4}$ & $\psi_{43}/\psi_{44}$\\
$$ & $$ & $$ & $$\\
\hline\cline{1-0}
$$ & $$ & $$ & $$\\
$\alpha_{1}+\alpha_{2}$ & $\left(
\begin{array}{c}
 1\\
 1\\
 -1
 \end{array}
\right)$ & $\omega_{1}-\omega_{3}$ & $\psi_{41}/\psi_{43}$\\
$$ & $$ & $$ & $$\\
\hline\cline{1-0}
$$ & $$ & $$ & $$\\
$\alpha_{2}+\alpha_{3}$ & $\left(
\begin{array}{c}
 -1\\
 1\\
 1
 \end{array}
\right)$ & $\omega_{2}-\omega_{4}$ & $\psi_{42}/\psi_{44}$\\
$$ & $$ & $$ & $$\\
\hline\cline{1-0}
$$ & $$ & $$ & $$\\
$\alpha_{1}+\alpha_{2}+\alpha_{3}$ & $\left(
\begin{array}{c}
 1\\
 0\\
 1
 \end{array}
\right)$ & $\omega_{1}-\omega_{4}$ & $\psi_{41}/\psi_{44}$\\

\hline
\end{tabular}
}
\end{center}

We can construct the invariant functions of the Toda field taking the products of the rational functions in the last columns from the Table 2 and the Table 3 corresponding to the identical roots.

\paragraph{6-dimensional representation on $\wedge^{2} V$}
In this case we choose the following basis
$$
\begin{array}{c}
f_{1}=e_{1} \wedge e_{2}, \ f_{2}=e_{1} \wedge e_{3}, \ f_{3}=e_{1} \wedge e_{4},\\
\ \\
f_{4}=e_{2} \wedge e_{3}, \ f_{5}=e_{2} \wedge e_{4}, \ f_{6}=e_{3} \wedge e_{4},\\
\ \\
v_{-}=f_{1}.
\end{array}
$$
Direct computations then yield:
$$
\begin{array}{c}
f_{j} = e_{k} \wedge e_{l},\\
\ \\
\rho_{6}(h_{\alpha_{i}}) \circ f_{j} = \rho_{4}(h_{\alpha_{i}}) \circ (e_{k} \wedge e_{l}) = (\omega_{k}(h_{\alpha_{i}})+\omega_{l}(h_{\alpha_{i}}))(e_{k} \wedge e_{l}),\\
\ \\
\omega^{\rho_{6}}_{j}=\omega_{k}(h_{\alpha_{i}})+\omega_{l}(h_{\alpha_{i}}),\\
\ \\
\rho_{6}(\Lambda) \circ f_{j} = \rho_{4}(\Lambda) \circ (e_{k} \wedge e_{l}) = \Lambda \circ (e_{k} \wedge e_{l})= (\lambda_{k}+\lambda_{l})(e_{k} \wedge e_{l}),\\
\ \\
\rho_{6}(Ad_{k}(\Lambda)_{h}) \circ f_{1} = (a_{11}+a_{22})(e_{1} \wedge e_{2}),\\
\end{array}
$$
and for functions $F^{\rho_{6}}_{j}$ we get the following expressions using the formulas (\ref{function1}) and (\ref{eq:compute1})
$$
\begin{array}{c}
F^{\rho_{6}}_{j}=M_{kl}
\ \\
M_{kl}=\psi_{1k}\psi_{2l}-\psi_{1l}\psi_{2k},
\ \\
M^{'}_{kl}=(-\lambda_{k}-\lambda_{l}+a_{11}+a_{22})M_{kl},
\end{array}
$$
Here $M$ denotes the corresponding top minor; i.e. $M_{ij}$ is the determinant of submatrix in the matrix $\Psi=\hat\rho(k)$ of the tautological representation spanned by the first two rows of $\Psi$ and by the $i$th and $j$th columns.

\newpage
Then the computations give the following result (the row $M$ describes the corresponding top minor):
\begin{center}
\mbox{Table 4}\\
\ \\
\small{
\begin{tabular}{|c|c|c|c|c|c|c|}
\hline\cline{1-0}
$$ & $$ & $$ & $$ & $$ & $$ & $$\\
$$  & $F_{1}$ & $F_{2}$ & $F_{3}$ & $F_{4}$ & $F_{5}$ & $F_{6}$\\

$$ & $$ & $$ & $$ & $$ & $$ & $$\\
\hline\cline{1-0}
$$ & $$ & $$ & $$ & $$ & $$ & $$\\
$M$ & $M_{12}$ & $M_{13}$ & $M_{14}$ & $M_{23}$ & $M_{24}$ & $M_{34}$\\

$$ & $$ & $$ & $$ & $$ & $$ & $$\\
\hline\cline{1-0}
$$ & $$ & $$ & $$ & $$ & $$ & $$\\
$\omega_{i} \in \mathfrak{a}^{\ast}$ & $\left(
\begin{array}{c}
 0\\
 -1\\
 1
 \end{array}
\right)$ & $\left(
\begin{array}{c}
 1\\
 -1\\
 1
 \end{array}
\right)$  & $\left(
\begin{array}{c}
 1\\
 0\\
 -1
 \end{array}
\right)$ & $\left(
\begin{array}{c}
 -1\\
 0\\
 1
 \end{array}
\right)$ & $\left(
\begin{array}{c}
 -1\\
 1\\
 -1
 \end{array}
\right)$ & $\left(
\begin{array}{c}
 0\\
 -1\\
 0
 \end{array}
\right)$\\

$$ & $$ & $$ & $$ & $$ & $$ & $$\\
\hline\cline{1-0}
$$ & $$ & $$ & $$ & $$ & $$ & $$\\
$\omega_{i}(\Lambda)$ & $\lambda_{1}+\lambda_{2}$ & $\lambda_{1}+\lambda_{3}$ & $\lambda_{1}+\lambda_{4}$ & $\lambda_{2}+\lambda_{3}$ & $\lambda_{2}+\lambda_{4}$ & $\lambda_{3}+\lambda_{4}$\\
$$ & $$ & $$ & $$ & $$ & $$ & $$\\

\hline
\end{tabular}
}
\end{center}

\newpage

\begin{center}
\mbox{Table 5}\\
\ \\
\small{
\begin{tabular}{|c|c|c|c|c|}
\hline\cline{1-0}
$$ & $$ & $$ & $$ & $$\\
$\alpha$  & $\alpha \in \mathfrak{a}^{\ast} $ & $\omega_{i}-\omega_{j}$ & $F_{\alpha_{ij}=\omega_{i}-\omega_{j}}=F_{i}/F_{j}$ & $F^{a}_{\alpha_{ij}}/F^{b}_{\alpha_{ij}}$\\

$$ & $$ & $$ & $$ & $$\\
\hline\cline{1-0}
$$ & $$ & $$ & $$ & $$\\
$\alpha_{1}$ & $\left(
\begin{array}{c}
 2\\
 -1\\
 0
 \end{array}
\right)$ & $\omega_{2}-\omega_{4}, \ \omega_{3}-\omega_{5}$ & $M_{13}/M_{23}, \ M_{14}/M_{24}$ & $\frac{M_{13}M_{24}}{M_{14}M_{23}}$\\

$$ & $$ & $$ & $$ & $$\\
\hline\cline{1-0}
$$ & $$ & $$ & $$ & $$\\
$\alpha_{2}$ & $\left(
\begin{array}{c}
 -1\\
 2\\
 -1
 \end{array}
\right)$ & $\omega_{5}-\omega_{6}, \ \omega_{1}-\omega_{2}$ & $M_{24}/M_{34}, \ M_{12}/M_{13}$ & $\frac{M_{24}M_{13}}{M_{34}M_{12}}$\\
$$ & $$ & $$ & $$ & $$\\
\hline\cline{1-0}
$$ & $$ & $$ & $$ & $$\\
$\alpha_{3}$ & $\left(
\begin{array}{c}
 0\\
 -1\\
 2
 \end{array}
\right)$ & $\omega_{2}-\omega_{3}, \ \omega_{4}-\omega_{5}$ & $M_{13}/M_{14}, \ M_{23}/M_{24}$ & $\frac{M_{13}M_{24}}{M_{14}M_{23}}$\\
$$ & $$ & $$ & $$ & $$\\
\hline\cline{1-0}
$$ & $$ & $$ & $$ & $$\\
$\alpha_{1}+\alpha_{2}$ & $\left(
\begin{array}{c}
 1\\
 1\\
 -1
 \end{array}
\right)$ & $\omega_{1}-\omega_{4}, \ \omega_{3}-\omega_{6}$ & $M_{12}/M_{23}, \ M_{14}/M_{34}$ & $\frac{M_{12}M_{34}}{M_{23}M_{14}}$\\
$$ & $$ & $$ & $$ & $$\\
\hline\cline{1-0}
$$ & $$ & $$ & $$ & $$\\
$\alpha_{2}+\alpha_{3}$ & $\left(
\begin{array}{c}
 -1\\
 1\\
 1
 \end{array}
\right)$ & $\omega_{4}-\omega_{6}, \ \omega_{1}-\omega_{3}$ & $M_{23}/M_{34}, \ M_{12}/M_{14}$ & $\frac{M_{23}M_{14}}{M_{34}M_{12}}$\\
$$ & $$ & $$ & $$ & $$\\
\hline\cline{1-0}
$$ & $$ & $$ & $$ & $$\\
$\alpha_{1}+\alpha_{2}+\alpha_{3}$ & $\left(
\begin{array}{c}
 1\\
 0\\
 1
 \end{array}
\right)$ & $\omega_{1}-\omega_{5}, \ \omega_{2}-\omega_{6}$ & $M_{12}/M_{24}, \ M_{13}/M_{34}$ & $\frac{M_{12}M_{34}}{M_{13}M_{24}}$\\
\hline
\end{tabular}
}
\end{center}
\ \\
The last column in this table gives the expression of the invariant function of the Toda field, equal to the ratio of two functions $F_\alpha^\rho$ (in this case we can take these functions from the same representation).

\section{Conclusions and remarks}
As we have seen in the previous sections, one can find a large number of infinitesimal symmetries of Toda field. Indeed, the size of commutative subalgebras in Lie algebras is usually rather big (\cite{Ma}); for instance in case $\mathfrak g=\mathfrak{sl}_{2k}$ the subalgebra spanned by the upper right (or bottom left) corner of the corresponding matrices is commutative. The dimension of this subalgebra is $k^2$, much larger than the dimension of the corresponding Cartan subalgebra. The analysis of representations of the group $G$ also show that the functions $F_\alpha^\rho$ with the necessary properties are usually available (the case of $G=SL_n$ is studied by methods similar to the ones we use in section 4).

Each vector field $\xi$, commuting with the Toda field, induces a 1-parameter group of diffeomorphisms $g_\xi^s$ on $K$ commuting with the Toda flow $g_\Lambda^t$. This means in particular that for any element $k\in K$, and all $t,s\in\mathbb R$ we have
\[
g_\xi^s(g_\Lambda^t(k))=g_\Lambda^t(g_\xi^s(k)).
\]
On the other hand the 1-parameter family $k(t)$ from proposition \ref{prop:flowt} can be written as $k(t)=g_\Lambda^t(k_0)$, thus we see that the formula
\[
L_\xi^{s}(t)=Ad_{g_\xi^s(g_\Lambda^t(k_0))}(\Lambda)
\]
gives a 1-parameter family of the solutions of the Toda system with initial condition $L_\xi^{s}(0)=Ad_{g^s_\xi(k_0)}(\Lambda)$.

Evidently, this construction can be easily generalised to the case of several commuting firelds. Unfortunately, because of the explicit dependence on the choice of the initial condition $k_0$ it is not clear, if this construction determines actual vector fields on the adjoint orbits of $\Lambda$. In order to check this we need to prove that the construction is invariant with respect to the right action on $K$ of the stabilizer group of $\Lambda$ (which is the cause of ambiguity in the choice of $k_0$). We can do this for the normal case, but if the real form is nonsplit, our methods yield only partial results, so we postpone the discussion of this question to a forthcoming paper.

Another interesting question is the relation of our construction with that of Reshetikhin and Schrader, see \cite{RS}. In their paper they considered a more geometric approach to the Toda system: in this work the authors consider the Toda system as a Hamiltonian system on the quotient space $G/K\cong\mathfrak p$ with Poisson structure induced from the usual Lie-Poisson structure on the group, induced from the classical $r$-matrix. In this context our construction would consist of finding analogues of the fields $\mathscr T^{\rho_\alpha}_\alpha$ on this quotient space $G/K$ and proving that they are Poisson (since the space $G/K$ is $1$-connected, this is equivalent to their being Hamiltonian). It is not difficult to extend the fields $\mathscr T^{\rho_\alpha}_\alpha$ to the froup $G$ from $K$; however, proving or disproving that they can be descended to the quotient is not so easy. We hope to answer this question in a forthcoming paper.

Also observe that the Toda field $\mathscr T^\Lambda$ on $K$ is known to be gradient (with respect to a specially designed Riemannian structure on $K$). On the other hand, its counterpart on $G/K=\mathfrak p$ is Hamiltonian. It is an interesting question, whether this phenomen can be accounted for by some analogue of the duality between different homogeneous spaces, similar to the duality between spherical and hyperbolic geometry.

\paragraph{Acknowledgements}
The work of Yu.B. Chernyakov and G.I. Sharygin was partly supported by grant RFBR-18-01-00461. The work of G.I Sharygin was partly supported by Simons foundation.

\end{document}